\documentclass{llncs}
\usepackage{amsmath}
\usepackage{amssymb}

\begin{document}
\title{Security Features of an Asymmetric Cryptosystem based on the Diophantine Equation Hard Problem and Integer Factorization Problem}

\author{M.R.K.Ariffin \inst{1,}\inst{2, [a]} \and M.A.Asbullah\inst{1,}\inst{2, [b]}  \and  N.A.Abu\inst{1,}\inst{3, [c]}}
\institute{Al-Kindi Cryptography Research Laboratory, Institute
for Mathematical Research, Universiti Putra Malaysia (UPM),
Selangor, Malaysia \and Department of Mathematics, Faculty of
Science, Universiti Putra Malaysia (UPM), Selangor, Malaysia \and
Faculty of Information Technology and Communication, Unitversiti
Teknikal Malaysia (UTeM), Melaka, Malaysia
\email{[a]rezal@math.upm.edu.my,}\email{[b]ma$\_$asyraf@putra.upm.edu.my,}\email{[c]nura@utem.edu.my}}

\maketitle

\begin{abstract}
\noindent The Diophantine Equation Hard Problem (DEHP) is a
potential cryptographic problem on the Diophantine equation
$U=\sum\limits_{i=1}^n {V_i x_{i}}$. A proper implementation of
DEHP would render an attacker to search for private parameters
amongst the exponentially many solutions. However, an improper
implementation would provide an attacker exponentially many
choices to solve the DEHP. The AA\,$_{\beta}$-cryptosystem is an
asymmetric cryptographic scheme that utilizes this concept
together with the factorization problem of two large primes and is
implemented only by using the multiplication operation for both
encryption and decryption. With this simple mathematical
structure, it would have low computational requirements and would
enable communication devices with low computing power to deploy
secure communication procedures efficiently.
\end{abstract}



\let\thefootnote\relax\footnotetext{Supported by Fundamental Research Grant Scheme $\#5523934$,
Ministry of Higher Education, MALAYSIA.}

\section{Introduction}
The discrete log problem (DLP) and the elliptic curve discrete log
problem (ECDLP) has been the source of security for cryptographic
schemes such as the Diffie Hellman key exchange procedure,
El-Gamal cryptosystem and elliptic curve cryptosystem (ECC)
respectively \cite{6}, \cite{10}. As for the world renowned RSA
cryptosystem, the inability to find the $e$-th root of the
ciphertext C modulo N from the congruence relation $C\equiv M^e
(\textrm{mod N})$ coupled with the inability to factor $N=pq$ for
large primes $p$ and $q$ is its fundamental source of security
\cite{11}. Recently, suggestions have been made that the ECC is
able to produce the same level of security as the RSA with shorter
key length. Thus, ECC should be the preferred asymmetric
cryptosystem when compared to RSA \cite{16}. Hence, the notion
``cryptographic efficiency" is conjured. That is, to produce an
asymmetric cryptographic scheme that could produce security
equivalent to a certain key length of the traditional RSA but
utilizing shorter keys. However, in certain situations where a
large block needs to be encrypted, RSA is the better option than
ECC because ECC would need more computational effort to undergo
such a task \cite{14}. Thus, adding another characteristic toward
the notion of ``cryptographic efficiency" which is it must be less
``computational intensive". As such, in order to design a
state-of-the-art public key mechanism, the above two
characteristics must be adhered to apart from other well known
security issues. In 1998 the cryptographic scheme known as NTRU
was proposed with better "cryptographic efficiency" relative to
RSA and ECC \cite{9}.  Much effort has been done to push NTRU to
the forefront \cite{8}.

The cryptographic scheme in this paper is based on what is defined
as the Diophantine Equation Hard Problem (DEHP). It is coupled
together with the well known integer factorization problem of two
large primes. The DEHP is a new form of cryptographic problem
based on the Diophantine equation of the form
$U=\sum\limits_{i=1}^n {V_i x_{i}}$. The authors propose that the
DEHP as outlined in this paper is also another cryptographic
problem that has secure cryptographic qualities coupled with the
above described ``cryptographic efficiency" qualities.

The layout of this paper is as follows. In Section 2, the
Diophantine Equation Hard Problem (DEHP) will be described. The
mechanism of the AA\,$_{\beta}$-cryptosystem will be detailed in
Section 3. Continuing in Section 4, will be discussion on the
security features of this cryptosystem. In Section 5 lattice based
attacks on the scheme is discussed. Section 6 will be devoted in
discussing the consequences of improper design utilizing the DEHP.
That is, the possibility of succumbing to a passive adversary
attack. The underlying principle and reduction proofs regarding
the intractability of the scheme is proposed in Section 7. A
numerical example of the scheme as well as an illustration of the
DEHP will also be given in this section. Finally, we conclude the
paper by comparing ``cryptographic efficiency" characteristics
against RSA,ECC and NTRU schemes in Section 8.

\section{The Diophantine equation hard problem (DEHP)}
The DEHP is based upon the linear diophantine equation which is of
the form $U=\sum\limits_{i=1}^n {V_i x_{i}}$. The following
definitions would give a precise idea regarding the DEHP.

\begin{definition}
Let $U=\sum\limits_{i=1}^n {V_i x_{i}^{*}}$ where the integers $U$
and $\left\{ {V_i } \right\}_{i = 1}^n $ are known. We define the
sequence of integers $ \left\{ {x_{i}^{*}}\right\}_{i = 1}^n $ as
the preferred integers used to obtain $U$. The sequence $ \left\{
{x_{i}^{*}}\right\}_{i = 1}^n $ are particular elements from the
set of solutions of $U=\sum\limits_{i=1}^n {V_i x_{i}^{*}}$ that
contains infinitely many elements. The problem to determine the
sequence $ \left\{ {x_{i}^{*}}\right\}_{i = 1}^n $ is known as the
DEHP.
\end{definition}

\begin{definition}
From Definition 1, for $n=2, V_{1}=1$ and $V_{2}=1$ the DEHP is
known as the AA\,$_{\beta}$-DEHP-2 (see Section 7).
\end{definition}

\begin{definition}
The Diophantine equation given by $U=\sum\limits_{i=1}^n {V_i
x_{i}^{*}}$ is defined to be \textit{prf}-solved when the sequence
of integers $ \left\{ {x_{i}^{*}}\right\}_{i = 1}^n $ are found in
order to obtain $U$. The DEHP or the AA\,$_{\beta}$-DEHP-2 is
solved when $U$ is \textit{prf}-solved.
\end{definition}

\begin{example}
Let $x_{1}=6143959510671614040, x_{2}=6143959507200090613$ be the
preferred solutions for the equation
$12287919017871704653=x_{1}+x_{2}$ where $x_1$ and $x_{2}$ are
$2n$-bits long (i.e. this example $n=32$). An attacker would be
faced with the AA\,$_{\beta}$-DEHP-2 (see Section 7) of
determining the preferred integer $x_{1}=t$ in order to determine
the remaining preferred integer $x_{2}=12287919017871704653-t$
that form the \textit{prf}-solution set for the above Diophantine
equation. Since it is known that $x_{1}$ is 64-bits long, the
possible values of $t$ resides within the interval $(2^{63},
2^{64}-1)$. In other words, there are $2^{64}$ possible values
that $x_{1}$ might be.
\end{example}

\section{The AA\,$_{\beta}$-Cryptosystem}
\noindent We will now define parameters needed for the renewed
$AA_\beta$-cryptosystem. The communication model is between two
parties A (Along) and B (Busu).

\begin{definition}
The ephemeral secret keys for Along are three integers. The
integers $a_{1}, a_{2}$ and $a_{3}$ are $2n$-bits long. The
relation between the integers is:
\begin{equation}
a_{1}+a_{2}\equiv 0 (\textrm{mod }a_{1}-a_{2})
\end{equation}
and
\begin{equation}
a_{2}+a_{3}\equiv v (\textrm{mod }a_{1}-a_{2})
\end{equation}
where $v$ is $0.8125n$-bits long.
\end{definition}

\begin{definition}
Let $p$ and $q$ be two prime numbers of $n$-bit length. Along's
public keys are given by
\begin{equation}
e_{A1}=a_{1} + a_{2}=pq
\end{equation}
and
\begin{equation}
e_{A2}=a_{1} + a_{3}
\end{equation}
\end{definition}

\begin{definition}
Along's private key is given by
\begin{equation}
d_{A1}=a_{1}-a_{2}=p
\end{equation}
\begin{equation}
d_{A2}=v
\end{equation}
\end{definition}

\begin{definition}
Busu will generate two ephemeral session keys: $k_{1}$ and
$k_{2}$. The keys $k_{1}$ and $k_{2}$ are $\frac{n}{6}$-bits long.
\end{definition}

\begin{definition}
The message that Busu will relay to Along is a
$(\frac{4n}{5})$-bit integer $m$.
\end{definition}

\begin{definition}
Busu will produce the following ciphertext:
\begin{equation}
C=k_{1}e_{A1}+k_{2}e_{A2}+m
\end{equation}
\end{definition}

\begin{proposition}
$(C(\textrm{mod } d_{A1}))(\textrm{mod } d_{A2})=m$.
\end{proposition}

\begin{proof}
We begin with:
\begin{equation}
(C(\textrm{mod } d_{A1}))=k_{2}v+m
\end{equation}
because $k_{2}v+m<d_{A1}$. Then,
\begin{equation}
(k_{2}v+m(\textrm{mod } d_{A2}))=m
\end{equation}
because $m<d_{A2}$.$\Box$
\end{proof}

\subsection{The AA\,$_{\beta}$ - public key cryptography scheme}
We will now discuss the AA\,$_{\beta}$-cryptosystem. It is as
follows: the scenario is that Busu will send an encrypted message
to Along. Along will provide Busu with his public key pair
$e_{A1}$ and $e_{A2}$. Busu intends to send the integer plaintext
$P=m$ as in Definition 8. Busu will then proceed to generate the
ciphertext $C$. Then Busu transmits the ciphertext $C$ to Along.
Upon receiving the ciphertext from Busu, Along by Proposition 1,
can retrieve the integer plaintext $P=m$.

\section{Security Features}
In this section we will focus on the obvious objective of an
attacker. That is to retrieve the plaintext or the private key or
both. Discussion would begin by discussing the objective of trying
to obtain the plaintext from the ciphertext followed by the
objective to obtain the private key embedded within the public
key.

\subsection{To obtain the plaintext from the ciphertext}
As defined in Definition 9, the plaintext resides within $C$.
Thus, the attacker has to \textit{prf}-solve $C$ via the preferred
integers $k_{1}$ and $k_{2}$ the AA\,$_{\beta}$-DEHP-1 (see
Section 7) given by
\begin{equation}
C=k_{1}e_{A1}+k_{2}e_{A2}+m
\end{equation}
The ability to determine the keys $k_{1}$ or $k_{2}$ would infer
that the attacker has also the ability to determine $m$ in the
first instance.

\subsection{To obtain the private key from the public key via the Diophantine equations}
The attacker has to \textit{prf}-solve $e_{A1}$ and $e_{A2}$ via
the preferred integers $a_{1}, a_{2}$ and $a_{3}$ the
AA\,$_{\beta}$-DEHP-2 (see Section 7). In congruent with the
ability to obtain the plaintext from the ciphertext as discussed
above, the ability to determine the keys $a_{1}, a_{2}$ and
$a_{3}$ would infer that the attacker has also the ability to
determine $m$ in the first instance.

\section{Lattice based attacks}
In this section we put forward two possible attacks via lattices
and show that why such attacks will not yield any information
detrimental to the scheme.

\subsection{Attack with Coppersmith method in the univariate case}
We will reproduce Coppersmith's theorem for the benfit of the
reader.

\begin{theorem}(Coppersmith)
Let $N$ be an integer of unknown factorization, which has a
divisor $b\geq N^{\beta}$. Furthermore, let $f_{\beta}(x)$ be an
univariate, monic polynimial of degree $\delta$. Then we can find
all solutions $x_{0}$ for the equation $f_{\beta}\equiv 0
(\textrm{mod } b)$ with
$$
\mid x_{0}\mid \leq \frac{1}{2} N^{\frac{\beta^{2}}{\delta} - \epsilon}
$$
in polynomial time in $(log N, \delta, \frac{1}{\epsilon})$.
\end{theorem}

\begin{case}
We begin by observing $e_{A1}=pq$ where $p$ and $q$ are of equal
length. Suppose $p$ is prime integer that satisfies
$p>(pq)^{\beta}$. It is clear that $\beta=\frac{1}{2}$. Let us now
observe the polynomials $x-e_{A2}$ and $e_{A1}=pq$ which have a
small common root $v$ modulo $p$. By the polynomial
$f_{p}(x)=x^{2}-e_{A2}x+(pq)$ we have the parameter $\delta=2$.
The parameter $\frac{1}{2} N^{\frac{\beta^{2}}{\delta} -
\epsilon}$ is an $(\frac{n}{4})$-bit integer while the parameter
$v$ is a $0.8125n$-bit integer. Thus, the bound is much smaller
than the root.
\end{case}

\begin{case}
A more efficient method would be just to observe the polynomial
$f_{p}(x)=x-e_{A2}$. Hence, $\delta=1$. The parameter $\frac{1}{2}
N^{\frac{\beta^{2}}{\delta} - \epsilon}$ is an $(\frac{n}{2})$-bit
integer while the parameter $v$ is a $0.8125n$-bit integer. Thus,
the bound is still much smaller than the root.
\end{case}

\subsection{Gaussian heuristic}
We will look at the the lattice $L$ spanned by $(1, 0,
e1),(0,1,e2),(0,0,C)$. Observe that the vector $V=(k1,k2,-m)$ is
in $L$. If $V$ is short, then the LLL algorithm will be able to
detect $V$. This is critical since by the usage of the vector
$V=(k1,k2,-m)$ it is obvious that the length of m is dominant when
compared to k1 and k2 hence length of V is approximately m. And by
the above information m is certainly dominant in the vector
V=(k1,k2,-m). Now let us check whether V is really short or not.
The Gaussian heuristic for the lattice L is given by:
\begin{equation}
\sigma (L)=\sqrt{(\frac{3}{2 \pi e})}   C^{1/3}
\end{equation}
One can see that $\sigma (L)$ is approximately $(\frac{2n}{3})$-bits,
while the length of the vector $V$ is $(\frac{4n}{5})$-bits. The
Gaussian heuristic is much smaller than the length of the vector $V$.
Thus, the vector $V$ is not considered to be short and cannot be detected
by the LLL algorithm.

\section{Improper design via the DEHP}
It is important to note that, an improper design of an asymmetric
cryptosystem via the DEHP would lead to succesful passive
adversary attacks. To illustrate this fact, we will produce the
following two examples.

\subsection{A key exchange mechanism based on the DEHP}
Let Along and Busu utilize private 2 X 2 non-singular matrices
$\textbf{A}$ and $\textbf{B}$ respectively. A base generator
$\textbf{G}$ will be made public. It is a 2 X 2 singular matrix.
The parameter $\textbf{E}_{A}=\textbf{AG}$ and
$\textbf{E}_{B}=\textbf{GB}$ will be exchanged between Along and
Busu. Then Along will compute
$\textbf{EAB}=[\textbf{A}]\textbf{E}_{B}$, while Busu will compute
$\textbf{EBA}=\textbf{E}_{A}[\textbf{B}]$. Now both parties have
the same key (i.e. key exchange). If the assumption is that the
attacker has to obtain either $\textbf{A}$ or $\textbf{B}$ from
either $\textbf{E}_{A}$ or $\textbf{E}_{B}$ this would be the
DEHP, since $\textbf{G}$ is singular. However, an attacker could
still compute $\textbf{A}^{'} \neq \textbf{A}$ but
$\textbf{A}^{'}\textbf{G}=\textbf{AG}$ and as a result is able to
compute $\textbf{A}^{'}\textbf{E}_{B}=\textbf{EAB}$. Thus
rendering the scheme insecure. The following is a numerical
example.

\begin{example}
Let
\[
\textbf{G}  = \left( {\begin{array}{*{20}c}
   1 & 2   \\
   2 & 4  \\
\end{array}} \right),
\textbf{A} = \left( {\begin{array}{*{20}c}
   2 & 3  \\
   4 & 5  \\
\end{array}} \right),
\textbf{B} = \left( {\begin{array}{*{20}c}
   7 & 8  \\
   9 & 10  \\
\end{array}} \right)
\]
Along will generate
\[
\textbf{E}_{A} = \left( {\begin{array}{*{20}c}
   7 & 14  \\
   14 & 28  \\
\end{array}} \right)
\]
and Busu will generate
\[
\textbf{E}_{B} = \left( {\begin{array}{*{20}c}
   25 & 28  \\
   50 & 56  \\
\end{array}} \right)
\]
The shared key computed by both parties is
\[
\textbf{AGB} = \left( {\begin{array}{*{20}c}
   175 & 196  \\
   350 & 392  \\
\end{array}} \right)
\]
An attacker intercepting $\textbf{E}_{A}$ could construct the
matrix
\[
\textbf{A}^{'} = \left( {\begin{array}{*{20}c}
   7 & 0  \\
   14 & 0  \\
\end{array}} \right)
\]
It could be observed that
$\textbf{AGB}=\textbf{A}^{'}\textbf{GB}$. Hence, a passive
adversary attack has been successfully executed.
\end{example}

\subsection{Improper integer size}
Observe the equation given by
\begin{equation}
e_{A}=a_{1}+a_{2}g_{1}
\end{equation}
where $e_{A}$ and $g_{1}$ are public parameters. Let $g_{1}$ be of
length $2n$-bits, while the private parameters $a_{1}$ and $a_{2}$
are $n$-bits long.Because of this improper choice of size, one can
obtain
\begin{equation}
a_{2}=floor(\frac{e_{A}}{g_{1}})
\end{equation}

\section{The Underlying Security Principle}
We will now observe the underlying security principles that the
$AA_\beta$-cryptosystem is based upon.

\subsection{The $AA_\beta$-DEHP-1}
Determine the preferred integer either  $(k_{1}$ or $k_{2})$ such
that $m=C-k_{1}e_{A1} (\textrm{mod } e_{A2})$ or  $m=C-k_{2}e_{A2}
(\textrm{mod } e_{A1})$.

\subsection{The $AA_\beta$-DEHP-2}
Determine the preferred integers $(a_{1}, a_{2}, a_{3})$ belonging
to the public keys $e_{A1}$ and $e_{A2}$.

\subsection{The integer factorization problem}
Let $p$ and $q$ be two large primes. From $e_{A1}=a_{1} +
a_{2}=pq$ obtain $d_{A1}=p$.

\subsection{Security reduction}
\begin{proposition}
$AA_\beta$-DEHP-2 $\equiv_{T}$ Factoring $e_{A1}=pq$.
\end{proposition}
\begin{proof}
Let $\theta_{1}$ be an oracle that factors the product of primes.
Call $\theta_{1}(e_{A1})$ to obtain $p$ and $q$. Then we are able
to construct $a_{1}=\frac{p(q+1)}{2}$, $a_{2}=\frac{p(q-1)}{2}$
and $a_{3}=e_{A2}-a_{1}$. Hence, the preferred integers $(a_{1},
a_{2}, a_{3})$ are obtained Thus, $AA_\beta$-DEHP-2 $\leq_{T}$
Factoring $e_{A1}=pq$. Let $\theta_{2}$ be an oracle that obtains
the preferred integers $(a_{1}, a_{2}, a_{3})$. Then obtain
$p=a_{1}-a_{2}$ and $\frac{e_{A1}}{p}=q$. Thus, Factoring
$e_{A1}=pq \leq_{T}$ $AA_\beta$-DEHP-2. Hence, $AA_\beta$-DEHP-2
$\equiv_{T}$ Factoring $e_{A1}=pq$. $\Box$
\end{proof}

\begin{proposition}
Decryption $\leq_{T}$ Factoring $e_{A1}=pq$.
\end{proposition}
\begin{proof}
Let $\theta_{1}$ be an oracle that factors the product of primes.
Call $\theta_{1}(e_{A1})$ to obtain $p$ and $q$. Then determine
$v\equiv e_{A2} (\textrm { mod} p)$. Now, decryption can occur.$\Box$
\end{proof}

\subsection{Indistinguishability}
\begin{proposition}
The $AA_\beta$ public key cryptosystem is IND-CPA.
\end{proposition}
\begin{proof}
The $AA_\beta$ public key cryptosystem is a probabilistic
cryptosystem. A probabilitic encryption scheme is IND-CPA
\cite{16}. Thus the $AA_\beta$ public key cryptosystem is IND-CPA.
$\Box$
\end{proof}

\subsection{Example}
We will now provide a clear numerical illustration of the
$AA_\beta$-cryptosystem for $n=32$-bits. Along will generate the
following secret keys: $a_{1}=6143959510671614040,$
$a_{2}=6143959507200090613$, $a_{3}=5113460585870913605$ and $v=66857602$.
Along's public keys are $e_{A1}=12287919017871704653$ and
$e_{A2}=11257420096542527645$. Observe that $e_{A1}$ is product of
two 32-bit primes ($p=3471523427$ and $q=3539633039$). Along's
private keys are $d_{A1}=3471523427$ and $d_{A2}=66857602$. In the
meantime Busu will generate $k_{1}=33$ and $k_{2}=32$. The message is
$M=39152991$. The ciphertext generated
by Busu is $C=765738770679166291180$. Finally,
$(C (\textrm{mod } d_{A1}))(\textrm{mod } d_{A2})=39152991$.$\square$

\section{Conclusion}
The $AA_\beta$-cryptosystem has the capacity to become a novel
public key cryptosystem whose hard mathematical problem is based
upon the difficulty of the DEHP and the integer factorization
problem of two large primes. Just like the RSA, where the $e$-th
root problem is considered much more difficult than factoring the
product of primes, the DEHP could also be considered much more
difficult than factoring the product of primes (due to the
exponential number of possibilities for the private parameters).
The minimum key length for optimum security should be set to
$n=512$-bits. On another note, it is known that the implementation
of RSA and ECC is $O(n^3)$ operations where $n$ is the length of
the message block \cite{5},\cite{8},\cite{17}. By this fact we can
have the following table of comparison.
\\

\begin{center}
    \begin{tabular}{|c|c|c|c|}
        \hline
        Algorithm & Encryption Speed & Decryption Speed & Expansion  \\ \hline
        RSA   & $O(n^2)$ & $O(n^3)$ & 1 - 1\\ \hline
        ECC  & $O(n^3)$ & $O(n^3)$ & 1 - 2 (2 parameter ciphertext)\\ \hline
        NTRU & $O(n^2)$ & $O(n^2)$ & varies\\ \hline
        $AA_\beta$ & $O(n^2)$ & $O(n^2)$ & 1 - 2.7\\ \hline
    \end{tabular}
\newline \\ \textbf{Table 2} \text{Encryption / decryption speed and message expansion table for
message block of length $n$}
\end{center}


One can also note another advantage. That is, since encrypt and
decrypt procedures are the basic arithmetic operation of
multiplication, the scheme could encrypt messages of large block
size with ease. As a result this algorithm is advantageous
relative to RSA or ECC (because of better speed) and ECC (because
of less computational effort to encrypt/decrypt messages of large
block size).

\section*{Acknowledgments}The authors would like to thank Yanbin Pan
of Key Laboratory of Mathematics Mechanization Academy of
Mathematics and Systems Science, Chinese Academy of Sciences
Beijing, China and Gu Chunsheng of School of Computer Engineering,
Jiangsu Teachers University of Technology, Jiangsu Province, China
for valuable comments and discussion.


\begin{thebibliography}{[AA1]}
\bibitem [1] {1} M.~R.~K. Ariffin and N.~A. Abu, $AA_\beta$-cryptosystem: A chaos based public key cryptosystem, \emph{Int. Jour. Cryptology Research}. vol. 1, no. 2 (2009), pp. 149--163.
\bibitem [2] {2} AM.~R.~K. Ariffin, N.~A. Abu and A. Mandangan, Strengthening the $AA_\beta$-cryptosystem, \emph{Proc. Second International Cryptology Conference 2010}. (2010), pp. 16--26.
\bibitem [3] {3} S.~R. Blackburn, The Discrete Log Problem Modulo 1: Cryptanalyzing
the Ariffin - Abu cryptosystem, \emph{J. Mathematical Cryptology},
vol. 4, no. 2, (2010), pp. 193--198.
\bibitem [4] {4} CR. Bose, Novel Public Key Encryption Techniques Based on Multiple Chaotic Systems, \emph{Physic Review Letters}. vol. 95, issue 9 (2005).
\bibitem [5] {5} A.~E. Cohen and K.~K. Parhi, Implementation of Scalable Elliptic Curve Cryptosystem Crypto-Accelerators for GF(2m), \emph{Conference Record of the Thirty-Eighth Asilomar Conference on Signals, Systems and Computers 1}. (2004), pp. 471--477.
\bibitem [6] {6} W. Diffie and M.~E. Hellman, New Directions in Cryptography, \emph{IEEE Transactions on Information Theory}. vol. 22, no. 26 (1976), pp. 644--654.
\bibitem [7] {7} J. Hoffstein, J. Pipher and J.~H. Silverman, An Introduction to Mathematical Cryptography. \emph{New York: Springer}. (2008), pp. 352--358.
\bibitem [8] {8} J. Hermans \textit{et. al.}, Speed Records for NTRU, \emph{CT-RSA 2010, LNCS 5985}. (2010), pp. 73--88.
\bibitem [9] {9} J. Hoffstein, J. Pipher, J.~H. Silverman. NTRU: A Ring Based Public Key Cryptosystem in Algorithmic Number Theory (ANTS III) Lecture Notes in Computer Science 1423, \emph{Springer-Verlag, Berlin}. (1998), pp. 267--288.
\bibitem [10] {10} N. Koblitz, Elliptic Curve Cryptosystems, \emph{Math. Comp}. vol. 48, no. 177 (1987), pp. 203--209.
\bibitem [11] {11} R.~L. Rivest, A. Shamir and L. Adleman, A method for obtainning digital signatures and public key cryptosystems, \emph{Commun. ACM}. vol. 21, issue 2 (1978), pp. 120--126.
\bibitem [12] {12} B. Schneier, Key length in Applied Cryptography. \emph{New York: John Wiley \& Sons}. (1996), pp. 151--168.
\bibitem [13] {13} M. Scott, When RSA is better than ECC. (2008, November 15) [Online]. Available: \textsc{http://www.derkeiler.com/Newsgroups/sci.crypt/2008-11/msg00276.html}
\bibitem [14] {14} S.~S. Wagstaff, Cryptanalysis of Number Theoretic Ciphers, \emph{Divisibility and Arithmetic}. (2003), pp. 27--42.
\bibitem [15] {15} S. Vanstone, ECC holds key to next generation cryptography. (2006, March 18) [Online].Available:
\textsc{http://www.design-reuse.com/articles/7409/ecc-hold-key-to-next-gen-cryptography.html}
\bibitem [16] {16} J. Wolkerstorfer and W. Bauer, A PCI-Card for Accelerating Elliptic Curve Cryptography, \emph{Proceedings of Austrochip 2002, Graz, Austria, October 4}, (2002).
\end{thebibliography}
\end{document}